\documentclass[conference]{IEEEtran}
\IEEEoverridecommandlockouts
 \usepackage{lipsum,graphicx,multicol}
\usepackage{color,soul}
\usepackage{mathtools}
\usepackage{graphicx}
\usepackage{color,soul}
\usepackage{setspace}
\usepackage{multirow}
\usepackage{tablefootnote}
\usepackage[norelsize, linesnumbered, ruled, lined, boxed, commentsnumbered]{algorithm2e}
\usepackage[normalem]{ulem}
\usepackage{comment}
\usepackage[font=small,labelfont=bf]{caption}
\usepackage{subcaption}
\usepackage{bm, amsmath, amssymb, amsthm}
\usepackage{amsfonts}
\usepackage {amssymb,amsmath,xcolor}
\usepackage{booktabs}
\usepackage{cite}
\usepackage{array}
\usepackage[acronym]{glossaries}
\usepackage{cases}
\usepackage{cleveref}
\usepackage{nicematrix}

\newtheorem{definition}{Definition}
\def\BibTeX{{\rm B\kern-.05em{\sc i\kern-.025em b}\kern-.08em
    T\kern-.1667em\lower.7ex\hbox{E}\kern-.125emX}}
    
  \newtheorem{theorem}{Theorem}

\newtheorem{corollary}{Corollary}
\newtheorem{proposition}{Proposition}

\newtheorem{remark}{Remark}

\newcommand{\bieee}{\begin{IEEEeqnarray}{rCl}}
\newcommand{\eieee}{\end{IEEEeqnarray}}

\usepackage{amsthm,mdframed,calc}

\SetKwInput{KwInput}{Input}                
\SetKwInput{KwOutput}{Output}              

\begin{document}

\title{Basis-Spline Assisted Coded Computing: Strategies and Error Bounds}

\author{\IEEEauthorblockN{Rimpi Borah$^{*}$, J. Harshan$^{*}$ and V. Lalitha$^{\dagger}$}
\IEEEauthorblockA{$^{*}$Indian Institute of Technology Delhi, India,\\ $^{\dagger}$International Institute of Information Technology, Hyderabad, India}}

\maketitle

\begin{abstract}
Coded computing has emerged as a key framework for addressing the impact of stragglers in distributed computation. While polynomial functions often admit exact recovery under existing coded computing schemes, non-polynomial functions require approximate reconstruction from a finite number of evaluations, posing significant challenges. Consequently, interpolation-based methods for non-polynomial coded computing have gained attention, with Berrut approximated coded computing emerging as a state-of-the-art approach. However, due to the global support of Berrut interpolants, the reconstruction accuracy degrades significantly as the number of stragglers increases. To address this challenge, we propose a coded computing framework based on cubic B-spline interpolation. In our approach, server-side function evaluations are reconstructed at the master using B-splines, exploiting their local support and smoothness properties to enhance stability and accuracy. We provide a systematic methodology for integrating B-spline interpolation into coded computing and derive theoretical bounds on approximation error for certain class of smooth functions. Our analysis demonstrates that the error bounds of our approach exhibit a faster decay with respect to the number of workers compared to the Berrut-based method. Experimental results also confirm that our method offers improved accuracy over Berrut-based methods for various smooth non-polynomial functions.
\end{abstract}


\section{Introduction}

Distributed computing is an effective paradigm for executing large-scale computational tasks, wherein a master coordinates the computation by  distributing the dataset among several workers \cite{a1}. A major performance bottleneck in such systems is the presence of \emph{stragglers}, i.e., slow or unresponsive workers, which can significantly delay the overall computation. To address this challenge, coded computing has been proposed as a fault-tolerant approach by introducing structured redundancy into distributed computations. This idea has been successfully applied to problems such as distributed matrix multiplication \cite{c1,w1} and polynomial function computation \cite{w2,w3,b5}. However, many existing coded computing schemes such as Lagrange Coded Computing \cite{w2}, Analog Lagrange Coded Computing \cite{b5} are limited to polynomial computations and rely on exact interpolation, which can be numerically unstable over real-valued data \cite{b6} and often require a large number of workers. In contrast, \cite{f1} proposed a numerically stable coded computing framework based on Chebyshev polynomials, however, this approach remains restricted to polynomial computations and due to its reliance on exact interpolation, still requires a large number of workers as the degree of polynomial increases.

To address these challenges, Berrut Approximated Coded Computing (BACC) \cite{b3} was proposed as a numerically stable framework for approximately computing arbitrary functions. BACC employs Berrut rational functions for both data encoding and reconstruction, thereby extending existing coded computing frameworks beyond polynomial functions while providing provable approximation error guarantees. Although BACC improves numerical stability, its reconstruction procedure relies on global approximation, which limits the achievable approximation order as the number of stragglers increases. To alleviate this problem, \cite{z5} explored a learning-theoretic framework to coded computing, wherein the underlying encoding and reconstruction operations were posed as regression problems, and were subsequently solved using data-centric methods. Interestingly, \cite{z5} was shown to outperform \cite{b3} in terms of relative error when handling non-polynomial functions. Although \cite{z5} is the state-of-the art in terms of accuracy metric, it is accompanied by other implementation overheads such as dataset management and model-training. Thus, within the class of non data-centric methods, BACC is still the state-of-the-art, and therefore, there is a need to investigate new structured coded computing methods that outperform numerical accuracy of BACC while explicitly characterizing the associated trade-offs. 

\subsection{Contributions}
\label{subsec:contributions}

Basis splines (B-splines) have been extensively studied in approximation theory and numerical analysis as a powerful tool for interpolation and function approximation \cite{z9,z1,z2,b2,b7}. In particular, they possess several desirable properties, including local support, higher-order approximation capability, and strong numerical stability. Despite this rich theoretical foundation, spline-based methods have not been systematically explored in the context of coded computing, where existing approaches continue to rely predominantly on polynomial or rational interpolation schemes. To bridge this gap, in this work, we propose a B-spline assisted coded computing (BSCC) framework for distributed computation of arbitrary functions in the presence of stragglers. Unlike global interpolation-based schemes, BSCC employs B-spline basis functions to perform local, piecewise-polynomial approximation of the target function, enabling higher-order approximation accuracy and improved numerical stability. 

Our specific contributions are listed below: For coded computing frameworks, we formulate a reconstruction methodology based on cubic B-spline interpolation. Our methodology explicitly specifies the construction of the spline space, the use of clamped knot vectors, and the natural boundary conditions, resulting in a well-posed linear system for reconstruction. We analyze the structure of the resulting B-spline interpolation system, and demonstrate that the reconstruction can be carried out efficiently using banded linear solvers, with computational complexity that scales linearly with the number of workers (see Section \ref{sec:basic BSCC}). We establish explicit bounds on the operator norm associated with the proposed B-spline reconstruction and connect these bounds to the overall approximation error. This analysis yields a precise characterization of how the reconstruction error depends on the number of workers and stragglers. In particular, for a class of smooth non-polynomial functions, we show that the approximation error of BSCC exhibits an inverse-quadratic decay with respect to the number of workers (see Section \ref{sec:analytical results}). Finally, through experimental results, we show that BSCC outperforms BACC in terms of approximation accuracy for a broader class of smooth non-polynomial functions, however, with a marginal increase in the computational complexity during the reconstruction process (see Section \ref{sec:exp results}).






\section{Preliminaries on B-Splines}
\label{subsec:prelim}
In this section, we revisit some basic definitions and fundamental properties of B-spline functions used in the our construction and analysis. To define B-spline functions, we define knot vectors, which capture the support set of B-splines. 

\begin{definition}
\label{def:knot_sequence}
A knot vector is a finite non-decreasing sequence of real numbers, denoted by
$\mathbf{t}=\{t_i\}_{i=0}^{m-1}$, such that $t_0 \leq t_1 \leq \cdots \leq t_{m-1}$,
where $m\in\mathbb{N}$. The elements of $\mathbf{t}$ are called knots.
\end{definition}

\begin{definition}
\label{def:cox de boor}
Let $p\ge 0$, and let $\mathbf{t}=\{t_j\}_{j=0}^{m-1}$ be a knot vector. Provided that $m\geq p+2$, the B-spline basis functions of degree $p$, denoted by $\{B_{j,p}(z)\mid j=0,1,\ldots,m-p-2\}$,
are defined recursively over the knot vector $\mathbf{t}$ \cite{z9} as follows:

\noindent For $p=0$,
\begin{small}
\[B_{j,0}(z)=
\begin{cases}
1, & z\in[t_j,t_{j+1}),\\
0, & \text{otherwise},
\end{cases}
\]\end{small} 

\noindent and for $p\geq 1$,
\[B_{j,p}(z)=\frac{z-t_j}{t_{j+p}-t_j} B_{j,p-1}(z)+\frac{t_{j+p+1}-z}{t_{j+p+1}-t_{j+1}} B_{j+1,p-1}(z).\]
If a denominator is zero, the corresponding term is defined to be zero.
\end{definition}
Each B-spline basis function $B_{j,p}(z)$ satisfies $B_{j,p}(z)\ge 0$ for all $z$ and $B_{j,p}(z)=0$ for $z\notin[t_j,t_{j+p+1})$. Consequently, each $B_{j,p}(z)$ has compact support and is nonzero only on the union of at most $p+1$ consecutive knot intervals. In the context of interpolation and approximation, B-spline basis functions are also defined over a special class of knot vectors called clamped knot vectors \cite{z1}, which are formally defined below.

\begin{definition}
\label{def:clamped_knot_vector}
Let $[a,b]\subset\mathbb{R}$ be a compact interval and let $p\ge 0$ and $n> p$ be integers. Let $ a=t_0 < t_1 < \cdots < t_{n-1}=b$ be a strictly increasing sequence of real numbers. A knot vector $\mathbf{t}$ is said to be \emph{$(p+1)$-clamped} on $[a,b]$ if it is defined as
\[\mathbf{t}=[\underbrace{t_{0},\ldots,t_{0}}_{p+1},
t_1,\ldots,t_{n-2},
\underbrace{t_{n-1},\ldots,t_{n-1}}_{p+1} ].\]
\end{definition}

\begin{definition}
\label{def:spline_space_bspline}
For $p\!\geq \!0$ and $m\!\geq\! p\!+\!2$, let $\{B_{j,p}(z)\mid j=0,1,\ldots,m\!-\!p\!-\!2\}$ denote the B-spline basis functions of degree $p$ defined with respect to a knot vector $\mathbf{t}=\{t_i\}_{i=0}^{m-1}$.  Using these basis functions a spline space denoted by $\mathcal{S}_{p,\mathbf{t}}$, can be generated as\cite{z2}

\begin{small}
\begin{equation}
\mathcal{S}_{p,\mathbf{t}}=\bigg\{\sum_{j=0}^{m-p-2} C_{j,p}\, B_{j,p}(z)\; \forall \;C_{j,p}\in\mathbb{R}
\bigg\}.
\end{equation}
\end{small}
\end{definition}

The B-spline basis functions $\{B_{j,p}(z)\mid j=0,1,\ldots,m\!-\!p\!-\!2\}$ satisfy the partition of unity property, i.e., $\sum_{j=0}^{m-p-2} B_{j,p}(z)=1$, over its support set.
\section{B-Spline Assisted Coded Computing}
\label{sec:basic BSCC}
We propose a B-spline assisted coded computing (BSCC) framework for approximate computation of an arbitrary function $f(\cdot)$ over a given input dataset. Our setup consists of a master and $N$ workers indexed by the set $\mathcal{W}=\{\mathcal{W}_{0},\mathcal{W}_{1},\ldots,\mathcal{W}_{N-1}\}$. We assume that each worker communicates exclusively with the master through a dedicated communication link, and no communication is permitted among workers. The input dataset is stored at the master, while all workers are assumed to have prior knowledge of the function to be evaluated. Among the workers, we assume the presence of $S$ number of stragglers, for $1 \leq S<N-2$, which may fail to return their computation results within the prescribed time. Let $\mathbf{X}=(\mathbf{X}_0,\ldots,\mathbf{X}_{K-1})$ denote the dataset held by the master, where $\mathbf{X}_j \in \mathbb{R}^{m\times n}$ for all $j \in [K]$, with $[K]\triangleq\{0,1,2,\ldots,K-1\}$. The primary objective of BSCC framework is to approximately evaluate an arbitrary function $f:\mathbb{R}^{m\times n}\rightarrow\mathbb{R}^{m\times n}$ on the matrices in $\mathbf{X}$ in a decentralized manner while guaranteeing tolerable approximation errors against $S$ number of stragglers. In particular, since the function $f(\cdot)$ may be non-polynomial, BSCC aims to enable numerically stable distributed computation of $f(\cdot)$ on $\mathbf{X}_j$ for all $j\in[K]$ with bounded approximation error; that is, if $\mathbf{Y}_j$ denotes the output reconstructed from the distributed computation corresponding to $\mathbf{X}_j$, then  $\mathbf{Y}_j \approx f(\mathbf{X}_j)$ for all $j\in[K]$ within a tolerable accuracy loss. Within this setting, we employ a generalized encoding method available in the literature and propose a B-spline assisted numerically stable reconstruction procedure at the master.\footnote{Coded computing under additional constraints such as data privacy and Byzantine adversaries are well studied. In this work, we restrict our attention to resilience against stragglers. Privacy guarantees and Byzantine robustness are straightforward to generalize using existing techniques.} 
\subsection{Encoding }
\label{subsec:encoding}
We describe the encoding method, specifically highlighting how the dataset $\mathcal{X}$ is mapped into distributed shares for the workers. To encode the dataset, the master may employ any suitable interpolatory basis functions $\{\Phi_j(\cdot)\}$, $j\in[K]$, such as the Lagrange basis \cite{w2} or the Berrut basis functions \cite{b3}. These basis functions satisfy the cardinality property, namely $\Phi_j(\beta_i)=1$ for $j=i$ and $\Phi_j(\beta_i)=0$ for $j\neq i$, for the underlying $K$ ordered non-decreasing encoding points $\{\beta_{i} \in \mathbb{R}~|~i = 0, 1, \ldots, K-1\}$. Using this property, the dataset $\mathcal{X}$ is encoded as a function of an evaluation variable $z$ as 
\begin{equation}
\label{eq:u(z)_encoding}
u(z) = \sum_{j=0}^{K-1} \mathbf{X}_j \Phi_j(z),
\end{equation}
where $\{\mathbf{X}_j\}_{j=0}^{K-1} \subset \mathcal{X}$ denote the data blocks and $\{\Phi_j(\cdot)\}$ are the chosen basis functions. 

\subsection{Distribution of Shares among the Workers}
\label{subsec:distribution of shares}

Once \(u(z)\) is constructed as in \eqref{eq:u(z)_encoding}, the master selects a set of \(N\) distinct and ordered evaluation points \(\{\alpha_i\}_{i=0}^{N-1} \subset \mathbb{R}\), such that \(\alpha_0 < \alpha_1 < \cdots < \alpha_{N-1}\). For instance, these points could be chosen as Chebyshev nodes of the first kind or second kind. Using these evaluation points, the master generates the encoded shares by evaluating the encoding function \(u(z)\) at each \(\alpha_i\), resulting in the set of function values \(\{\mathbf{Y}_i\}_{i=0}^{N-1}\), where $\mathbf{Y}_i = u(\alpha_i), \quad i\in[N].$ Finally, the master distributes the encoded shares to the workers by assigning the share \(\mathbf{Y}_i\) to worker \(\mathcal{W}_i\), for \(i\in[N]\). 

\subsection{Computation at the Workers}
\label{subsec:computation at worker}

Once the worker $\mathcal{W}_{i}$ receives its evaluation $\mathbf{Y}_{i}$, it computes $f(\mathbf{Y}_{i})=f(u(\alpha_{i}))$, and returns it to the master. In contrast, the stragglers may not return these values within the prescribed deadline. 

\subsection{Function Reconstruction using B-spline}
\label{subsec:our framework cubic}
In the absence of stragglers, i.e., when $S=0$, the master receives $N$ function evaluations $\{ f(u(\alpha_i)) \mid i=0,1,\ldots,N-1 \}$ from $N$ workers, corresponding to the ordered evaluation points $\{\alpha_i\}_{i=0}^{N-1}$. The objective of the master is to reconstruct an approximation of the composite function $f(u(z))$ over the interval spanned by these points. Since the function $f(\cdot)$ is arbitrary and may be non-polynomial, exact interpolation~\cite{b5} is generally not feasible or numerically stable. Consequently, we approximate $f(u(z))$ using a B-spline function $r_{B-spline}(z) \approx f(u(z))$, where 
\begin{equation}
\label{bspline-approx}
r_{B-spline}(z) = \sum_{j} C_{j,p}\, B_{j,p}(z),
\end{equation}
where $\{ B_{j,p}(z)\}$ are B-spline basis functions of degree $p$, and the coefficients $\{C_{j,p}\}$ are obtained using an appropriate interpolation method using the $N$ function evaluations $\{ f(u(\alpha_i)) \mid i=0,1,\ldots,N-1 \}$. Using \eqref{bspline-approx}, the desired function $f(\mathbf{X}_{j})$ are reconstructed by evaluating the spline at the points $\{\beta_{j}\}_{j=0}^{K-1}$
as $\mathbf{Y}_{j} = r_{B-spline}(\beta_{j})$. Note that these are approximate evaluations, and the accuracy of this reconstruction depends on the reconstruction accuracy of $r_{B-spline}(z)$ in \eqref{bspline-approx}. In the next section, we present our interpolation method to choose the set $\{C_{j,p}\}$ using cubic B-splines, i.e., $p = 3$.\footnote{In approximation theory, it is well known that among various degrees of B-spline, cubic B-splines achieve a good balance between the accuracy of approximation and computational complexity for interpolation.}

\subsection{Cubic B-Spline Interpolation for BSCC}
\label{subsec:framework:clamped B-spline}
With the use of cubic B-splines, i.e., $\!p =\!3$, let \(\{B_{j,3}(z)\}_{j=0}^{N+1}\) denote the set of cubic B-spline basis functions, where the spline space $\mathcal{S}_{3,\mathbf{t}}$ is constructed over the $4$-clamped knot vector $\mathbf{t}$ given by
\begin{equation}
\label{eq:cubic_knots}
\mathbf{t}=\bigl[\underbrace{\alpha_0,\ldots,\alpha_0}_{4},
\alpha_1,\alpha_2,\ldots,\alpha_{N-2},
\underbrace{\alpha_{N-1},\ldots,\alpha_{N-1}}_{4}\bigr],
\end{equation}
where \(\{\alpha_i\}_{i=0}^{N-1}\) are the distinct and ordered evaluation points for BSCC. Given \(N\) function evaluations \(\{ f(u(\alpha_i)) \mid i=0,1,\ldots,N-1 \}\), we approximate the composite function \(f(u(z))\) by a cubic spline \(s(z)\in\mathcal{S}_{3,\mathbf{t}}\) of the form
\begin{equation}
\label{eq:cubic_spline_expansion}
s(z)=\sum_{j=0}^{N+1} C_{j,3}\, B_{j,3}(z),
\end{equation}
where \(\{C_{j,3}\}_{j=0}^{N+1}\) are the spline coefficients that need to be computed. Towards that direction, the interpolation conditions
$s(\alpha_i)=f(u(\alpha_i)), \quad i=0,\ldots,N-1,$ lead to a system of linear equations given by $\mathbf{B}_{rect}\mathbf{c} = \mathbf{f},$ where the rectangular B-spline matrix
\(\mathbf{B}_{\mathrm{rect}}\in\mathbb{R}^{N\times(N+2)}\) given by 
\begin{equation}
\label{eq:B_rect}
\begin{aligned}
&
\mathbf{B}_{rect} = \begin{bmatrix}
B_{0,3}(\alpha_0)  & B_{1,3}(\alpha_0)  & \cdots & B_{N+1,3}(\alpha_0) \\
B_{0,3}(\alpha_1)  & B_{1,3}(\alpha_1)  & \cdots & B_{N+1,3}(\alpha_1) \\
\vdots             & \vdots             & \ddots & \vdots \\
B_{0,3}(\alpha_{N-1}) & B_{1,3}(\alpha_{N-1}) & \cdots & B_{N+1,3}(\alpha_{N-1}) \\
\end{bmatrix}
\end{aligned}
\end{equation}
has full row-rank \(N\), and the vectors $\mathbf{c}$ and $\mathbf{f}$ are given by
\begin{equation}
\begin{aligned}
&
\mathbf{c} = \begin{bmatrix}
C_{0,3} \\
C_{1,3} \\
\vdots \\
C_{N+1,3}
\end{bmatrix},
\mathbf{f} = \begin{bmatrix}
f(u(\alpha_0)) \\
f(u(\alpha_1)) \\
\vdots \\
f(u(\alpha_{N-1}))\\
\end{bmatrix}.
\end{aligned}
\end{equation}
Consequently, the resulting linear system is underdetermined and admits infinitely many solutions. To obtain a unique solution, \(p\!-1\!=2\) additional constraints are required. Therefore, we impose the natural spline boundary conditions $s''(\alpha_0)=0 \quad \text{and} \quad s''(\alpha_{N-1})=0,$ where  $s''(z)$ denotes the double derivative of $s(z)$. 
Since $s''(z)=\sum_{j=0}^{N+1} C_{j,3}\, B''_{j,3}(z),$ these boundary conditions yield two additional linear constraints on the spline coefficients \cite{b7}. Evaluating them gives
\[\sum_{j=0}^{N+1} C_{j,3} B''_{j,3}(\alpha_0)=0\quad \text{and} \quad\sum_{j=0}^{N+1} C_{j,3} B''_{j,3}(\alpha_{N-1})=0.\]
By augmenting the rectangular system with these two boundary constraints, we obtain an $(N+2) \times (N+2)$ square B-spline matrix.
The resulting system of linear equations can be written as $\mathbf{B}\mathbf{c} = \bar{\mathbf{f}}$, where $\bar{\mathbf{f}} = [0 ~\mathbf{f}^{T} ~0]^{T}$ and 
\begin{equation}
\label{eq:square_bspline_matrix}
\mathbf{B}=\begin{bmatrix}
\mathbf{a} \\
\mathbf{B}_{\mathrm{rect}} \\
\mathbf{b}
\end{bmatrix},
\end{equation}

\noindent where the row vectors $\mathbf{a}, \mathbf{b} \in \mathbb{R}^{1 \times (N+2)}$are given by
\[\mathbf{a}=\bigl[B''_{0,3}(\alpha_0),\;B''_{1,3}(\alpha_0),\;\ldots,\;B''_{N+1,3}(\alpha_0)\bigr],\]
\[\mathbf{b}=\bigl[B''_{0,3}(\alpha_{N-1}),\;
B''_{1,3}(\alpha_{N-1}),\;\ldots,\;B''_{N+1,3}(\alpha_{N-1})\bigr].\]

The two additional zero entries in $\bar{\mathbf{f}}$ correspond to the imposed natural spline boundary conditions. Before solving this linear system, we next present results on the structure of \(\mathbf{B}\) and establish results on its invertibility as well as the computational complexity for its inversion.

\begin{proposition}
\label{prop:invertibility}
For cubic B-spline i.e., $p\!=\!3$, the B-spline matrix  $\mathbf{B}\!\in\! \mathbb{R}^{(N+2)\times (N+2)}$ in \eqref{eq:square_bspline_matrix} has the following properties:
(i)$\mathbf{B}$ is is banded, with at most $4$ non zero entries in each row. (ii) $\mathbf{B}$ is full rank, hence the associated linear system $\mathbf{B}\mathbf{c} = \bar{\mathbf{f}}$ admits a unique solution.
\end{proposition}
\begin{proof}
See Appendix \ref{proof:1} for the proof.
\end{proof}
\begin{proposition}
\label{prop:complexity}
For cubic B-splines, the matrix $\mathbf{B}\in \mathbb{R}^{(N+2)\times (N+2)}$ in \eqref{eq:square_bspline_matrix} has upper and lower bandwidth of $2$. Subsequently, the associated linear system can be solved using banded LU decomposition with computational complexity $\mathcal{O}(8(N+2))$\cite{b90}, which is smaller than that of a dense matrix.
\end{proposition}

Overall, after obtaining the spline coefficients $\{C_{j,3}\}_{j=0}^{N+1}$, the desired function values are reconstructed as explained in Section \ref{subsec:our framework cubic}. In the rest of the section, we discuss the modifications required for interpolation in the presence of stragglers, i.e., when $S>0$. Let $\mathcal{F}\!\!=\!\!\{l_1,l_2,\ldots,l_M\}$ denote the index set of the $M\!\!=\!\!N-S$ non-stragglers. Let $\{\mathbf{R}_{l_1},\mathbf{R}_{l_2},\ldots,\mathbf{R}_{l_M}\}$ denote the corresponding results received at the master, where $\mathbf{R}_{l_i}=f\!\left(u(\alpha_{l_i})\right)$ for $l_i\in\mathcal{F}$. Let $\{\alpha_{l_{i}}\}_{i=0}^{M-1}\subset\{\alpha_i\}_{i=0}^{N-1}$ denote the ordered set of evaluation points corresponding to the non-stragglers in the set $\mathcal{F}$. To generate the B-spline basis functions, the master constructs the associated $4$-clamped knot vector from evaluation points of the non-stragglers set. Subsequently, $\mathbf{B}_{rect}$ defined in \eqref{eq:B_rect} reduces to the dimension $M\times(M+2)$, where the basis functions $\{B_{i,3}(z)\}_{i=0}^{M+1}$ are evaluated at the points $\{\alpha_{l_{i}}\}_{i=0}^{M-1}$. Further, after imposing the boundary conditions on $\mathbf{B}_{rect}$, we obtain $\mathbf{B}$ defined in \eqref{eq:square_bspline_matrix} of dimension $(M+2)\times(M+2)$. Finally, solving the resulting linear system yields the spline coefficients, which are subsequently used to recover the desired values $\{f(\mathbf{X}_j)\}_{j=0}^{K-1}$.
Similar to BACC, there is no strict recovery threshold for BSCC as well. The master reconstructs the function using the evaluations from the non-stragglers, with accuracy improving as more evaluations are received.

\section{Error Bounds for BSCC}
\label{sec:analytical results}
As discussed in Section \ref{subsec:our framework cubic}, the reconstruction accuracy of $\{f(\mathbf{X}_{j})\}_{j=0}^{K-1}$ depends on the accuracy with which $r_{B-spline}(z)$ approximates $f(u(z))$. Therefore, we are interested in deriving bounds on the approximation error $||r_{B-spline}(z)-f(u(z))||_{\infty}$, where $||\cdot||_{\infty}$ denotes the infinity norm over the support set. Henceforth, we use $g(z)$ to refer to $f(u(z))$ and $T g(z)$ to refer to $r_{B-spline}(z)$, assuming that $T$ is the underlying B-spline interpolation operator that approximates $g(z)$ to $Tg(z)\in\mathcal{S}_{3,\mathbf{t}}$. Towards evaluating the interpolator operator $T$, we recollect this definition from \cite{z9}.

\begin{definition}
\label{def:operator_norm}
Let $C[a,b]$ represent the set of all continuous functions defined over $[a,b]$. Further, for any function $g(z)\in C[a,b]$, let $Tg(z)\in \mathcal{S}_{3,\mathbf{t}}$ denote the spline function obtained by applying the operator $T$ to $g(z)$. Then the operator norm of $T$ is defined as

\begin{equation}
\label{eq:opt_norm_Def}
\|T\|:=\sup_{g(z)\neq 0}\frac{\|Tg(z)\|_{\infty}}{\|g(z)\|_{\infty}}.
\end{equation}
\end{definition}
For the proposed BSCC interpolation  with cubic B-splines, we derive an upper bound on the operator norm $||T||$, which is presented in the following proposition.
\begin{proposition}
\label{lem:lebesgue}
For a given $N$, consider the BSCC setting when $S=0$. Let $T$ denote our spline interpolation operator that maps the composite function $g(z)$ to a spline function $Tg(z)\in \mathcal{S}_{3,\mathbf{t}}$. Then, $\|T\|$ is upper bounded as
\begin{equation}
\|T\|\leq \|[\mathbf{B}^{-1}]_{\mathrm{sub}}\|_{\infty},
\end{equation}
where $[\mathbf{B}^{-1}]_{\mathrm{sub}}\!\in\!\mathbb{R}^{(N+2)\times{N}}$ denotes the submatrix of $\mathbf{B}^{-1}$ after excluding the first and the last columns, where $\mathbf{B}$ is given in \eqref{eq:square_bspline_matrix}.
\end{proposition}
\begin{proof}
      See Appendix \ref{proof:2} for the proof.
\end{proof}
For a given set of distinct ordered evaluation points $\{\alpha_i\}_{i=0}^{N-1}$, the above upper bound $\|[\mathbf{B}^{-1}]_{\mathrm{sub}}\|_{\infty}$ can be computed using \eqref{eq:square_bspline_matrix} in a straightforward manner. However, in the following theorem, we derive an upper bound on $\|[\mathbf{B}^{-1}]_{\mathrm{sub}}\|_{\infty}$ only as a function of the knot spacings.


\begin{theorem}
\label{thm:lebegue_bound}
For a given $N$, a knot vector $\mathbf{t}$, when $S\!\!\!=\!\!\!0$, let $h_{\max}\!\!\!=\!\!\!\max_{\substack{0 \leq i \leq N-2\\}}\left( t_{i+1} - t_i \right)$ and $h_{\min}\!\!=\!\!\min_{\substack{0 \leq i \le N-2}}\left( t_{i+1} - t_i \right)$ represent the maximum and minimum knot spacing between consecutive distinct knots in the knot vector. Then
\[\|[\mathbf{B}^{-1}]_{\mathrm{sub}}\|_{\infty}\leq C_{1} (N+2)\frac{h_{\max}}{h_{\min}},\]
where $C_{1}>0$ is a constant, which depends only on the degree of the spline.
\end{theorem}
\begin{proof}
      See Appendix \ref{proof:3} for the proof.
\end{proof}

Further, using the above upper bound on $\|[\mathbf{B}^{-1}]_{\mathrm{sub}}\|_{\infty}$, we use the results from \cite{z9} to present an upper bound on the approximation error of BSCC.
\begin{theorem}
\label{th:first prelim}
For a given $N$, consider the BSCC scheme when $S=0$. Then, for every function $g(z) \in C[a,b]$, the interpolation error of BSCC satisfies
\begin{align}
\label{eq:app_error_1}
\|g(z) - Tg(z)\|_{\infty}
&\leq \Bigl(1 + C_{1}(N+2) \frac{h_{\max}}{h_{\min}}\Bigr)D_{S_{3, t}}(g(z)),
\end{align}
where $D_{S_{3, t}}(g(z)) = \inf_{s(z) \in S_{3,\mathbf{t}}} \|g(z) - s(z)\|_{\infty}$, and all other constants are as defined in Theorem \ref{thm:lebegue_bound}.
\end{theorem}

\begin{proof}
    See Appendix \ref{proof:4} for the proof.
\end{proof}
While the above theorem gives a bound on the approximation error of BSCC for an arbitrary function $g(z)$, there is no simpler representation of $D_{S_{3, t}}(z)$ in general. However, for a special class of functions, i.e., when $g(z)\in C^{4}[a,b]$, where $C^{4}[a,b]$ denotes the set of all functions over $[a,b]$ whose fourth derivative is also continuous, we use the results in \cite[Chapter XII, page 149]{z9} to bound $D_{S_{3, t}}(g(z))$, and subsequently derive an upper bound on \eqref{eq:app_error_1}.
\begin{corollary}
\label{th:second prelim}
For the BSCC setting considered in Theorem \ref{th:first prelim}, if $g(z)\in C^{4}[a,b]$, then there exists a constant $C>0$, depending only on the degree of spline, such that 

\begin{small}
\begin{equation}
\label{eq:app_error_2}
\|g(z) - Tg(z)\|_{\infty} \leq C\bigg(1 +  C_{1} (N+2)\frac{h_{\max}}{h_{\min}}\bigg)\,h_{max}^{4} \|g''''(z)\|_\infty,
\end{equation}
\end{small}

\noindent where $g''''(z)$ represents the fourth derivative of $g(z)$ in $[a, b]$, and all other constants are as defined in Theorem \ref{th:first prelim}.
\end{corollary}


\noindent All the above results are presented when $S = 0$. However, with $S > 0$, Theorem \ref{th:first prelim} and Corollary \ref{th:second prelim} continue to hold with the exception that $N$ must be replaced by $N-S$ and the parameters $h_{max}$ and $h_{min}$ must be computed using the evaluation indices of non-stragglers. 

\subsection{Error Bounds Comparison with BACC}
\label{subsec:BACC_BSCC}

Although BSCC allows any choice of evaluation points $\{\alpha_i\}_{i=0}^{N-1}$ over $[a,b]$, one can select $\{\alpha_i\}_{i=0}^{N\!-\!1}$ as Chebyshev nodes of the second kind,
defined by $\alpha_i= \cos\!\left(\frac{i\pi}{N}\right)$, for $i\in[N]$. In such a case, with $S$ stragglers, the upper bound in \eqref{eq:app_error_2} reduces to 

\begin{scriptsize}
\begin{equation}
\label{eq:bound_with_cheby_with_strag}
2 C\left(\sin^4\!\left(\frac{(S+1)\pi}{2N}\right)
+ \frac{C_{1}(N-S+2)}{h_{min}}\sin^5\!\left(\frac{(S+1)\pi}{2N}\right)\right)
\|g''''(z)\|_\infty, 
\end{equation}
\end{scriptsize}

\noindent where $h_{min}$ is independent of $S$, and is known to follow $\mathcal{O}(\frac{1}{N^2})$ behavior, for the class of Chebyshev nodes. We recall from \cite[Theorem 9]{b3} that the approximation error of BACC varies with $S$ and $N$ as

\begin{small}
\begin{equation}
\label{eq:Berrut_error}
\bigg(1 + \frac{(1+S)(3+S)\pi^{2}}{4}\bigg)\sin\!\left(\frac{(S+1)\pi}{2N}\right),
\end{equation}
\end{small}

\begin{figure}[ht!]
\centering
\includegraphics[scale = 0.23]{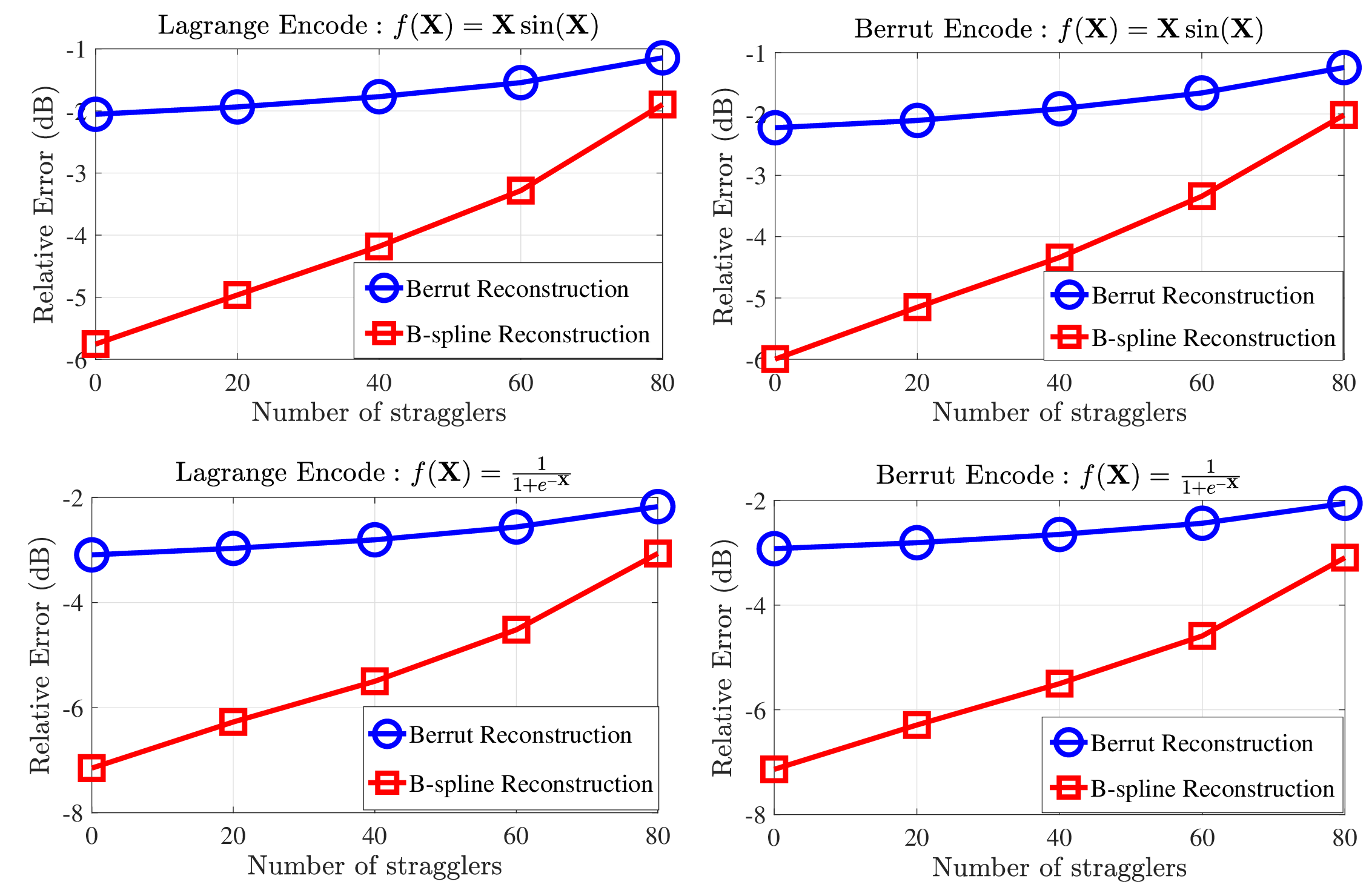}
\caption{Average relative error (in dB) of the BSCC and BACC schemes with Lagrange and Berrut encoding with $N\!=\!100$, $K=8$.}
\label{fig:plots}
\end{figure}

\noindent after excluding all other constants. Provided $g(z)\in C^{4}[a,b]$, using \eqref{eq:bound_with_cheby_with_strag} and \eqref{eq:Berrut_error}, we make the following inferences: With $S = 0$, the error bound of BACC varies as $\mathcal{O}(\frac{1}{N})$, whereas the error bound of BSCC varies as $\mathcal{O}(\frac{1}{N^{2}})$, thereby indicating faster decay of approximation error with the number of workers. For $0 < S << N$, in the regime of large $N$,  the error bound of BACC behaves as $\mathcal{O}(\frac{S^{3}}{N})$, whereas the error bound of BSCC varies as $\mathcal{O}(\frac{S^{5}}{N^{2}})$. This indicates that, for small $S$ and large $N$, the absolute error values of BSCC are expected to be much lower than that of BACC owing to the inverse-quadratic behavior with respect to $N$. Furthermore, along the lines of BACC, for a fixed $N$, we observed that the error bound of BSCC also grows with respect to $S$. In the next section, we attempt to validate these observations using experimental results for arbitrary functions.

\section{Experimental Results}
\label{sec:exp results}
In this section, we present experimental results on the accuracy of BSCC, and compare them with that of BACC \cite{b3}, for different numbers of stragglers. For the experiments, we use two different encoding methods: one based on the Lagrange basis functions \cite{w2} and the other based on the Berrut basis functions \cite{b3}. To demonstrate the results, we implement BSCC and BACC \cite{b3} to compute non-polynomial functions $f(\mathbf{X}) = \mathbf{X}\sin(\mathbf{X})$, $f(\mathbf{X})=\frac{1}{1+e^{-\mathbf{X}}}$. Here, $\mathbf{X}_{j} \in \mathbb{R}^{5\times 5}$ for $j \in [K]$, whose samples are drawn from uniform distribution. We use $\mathbf{Y}_{j} \approx f(\mathbf{X}_{j})$ to denote the approximate computation using BSCC or BACC for $j \in [K]$, and $\mathbf{Y}=f(\mathbf{X}_{j})$ to denote the centralized computation at the master without using BSCC or BACC . We define the relative error introduced by  BSCC and BACC with respect to the centralized computation as $e_{rel} \triangleq \frac{||\mathbf{Y}-\mathbf{Y'}||^{2}}{||\mathbf{Y}||^{2}},$ and finally compute the average relative error (in dB scale) over $10^3$ iterations, and present them in Fig. \ref{fig:plots} by varying $S$ with non-stragglers chosen uniformly random over $N$ in each iterations. 
The parameters used for the experiments are $N=100, K=8$, $p=3$, $\mathbf{X} \in \mathbb{R}^{40\times 5}$, $\{\beta_{j}\}_{j=0}^{K-1}$, and $\{\alpha_{i}\}_{i=0}^{N-1}$ are chosen from Chebyshev nodes of first and second kind, respectively. The plots in Fig.~\ref{fig:plots} confirm that BSCC achieves a remarkable improvement in approximation accuracy over BACC for non-polynomial functions, particularly when the number of stragglers is small, and offers improvement even when the number of stragglers increases. In other words, to maintain a relative error on the order of $10^{-2}$ or $10^{-3}$, BSCC can tolerate a substantially larger number of stragglers compared to BACC.

\section{Future work}
In this work, we have proposed a B-spline-assisted coded computing framework applicable to distributed computation of arbitrary functions. Our future work is to extend the proposed framework to incorporate Byzantine workers, develop robust encoding strategies based on B-splines, and adopt the proposed method to accelerate distributed computational tasks in a wide range of applications, including distributed learning.




\begin{appendix}
\subsection{Proof of Proposition \ref{prop:invertibility}}
\label{proof:1}
\begin{proof}
For cubic B-spline i.e., $p=3$, the structure of $\mathbf{B}^{(N+2)\times(N+2)}$ defined in \eqref{eq:square_bspline_matrix} is given by,

\begin{small}
\[\mathbf B =\left(
\begin{array}{ccccccc}
a_1 & -b_1 & d_1 & 0 & \cdots & 0 & 0 \\ \cline{1-7}

\multicolumn{1}{|c}{1} & 0 & 0 & 0 & \cdots & 0 & \multicolumn{1}{c|}{0} \\
\multicolumn{1}{|c}{0} & a_3 & b_3 & d_3 & \cdots & 0 & \multicolumn{1}{c|}{0} \\
\multicolumn{1}{|c}{0} & 0 & a_4 & b_4 & d_4 & \cdots & \multicolumn{1}{c|}{0} \\
\multicolumn{1}{|c}{\vdots} & \ddots & \ddots & \ddots & \ddots & \ddots & \multicolumn{1}{c|}{\vdots} \\
\multicolumn{1}{|c}{0} & \cdots & 0 & a_N & b_N & d_N & \multicolumn{1}{c|}{0} \\
\multicolumn{1}{|c}{0} & \cdots & 0 & 0 & 0 & 0 & \multicolumn{1}{c|}{1} \\ \cline{1-7}

0 & 0 & 0 & \cdots & a_{N+2} & -b_{N+2} & d_{N+2}
\end{array}
\right)\;\left.\vphantom{
\begin{array}{c}
1\\0\\0\\\vdots\\0\\0
\end{array}}
\right\}\;\mathbf B_{\mathrm{rect}}
\]\end{small}

\noindent where $a_j,b_j,d_j>0$ for all $j\in\{3,\ldots,N\}$. 
The first and last rows correspond to the natural spline boundary conditions for cubic B-spline, and satisfy $a_j-b_j+d_j=0$ for $j\in\{1,N+2\}$, which enforce $s''(\alpha_0)=\sum_{j=0}^{N+1} C_{j,3}\, B''_{j,3}(\alpha_0)=0$, and $s''(\alpha_{N-1})=\sum_{j=0}^{N+1} C_{j,3}\, B''_{j,3}(\alpha_{N-1})=0.$ The boxed interior rows i.e., from second to $N+1$-th row, constitute the rectangular matrix $\mathbf{B}_{\mathrm{rect}}$ and correspond to the interpolation conditions at the evaluation points $\{\alpha_{i}\}_{i=0}^{N-1}$. These rows satisfy $a_{j}+b_{j}+d_{j}=1$, for $j\in\{3,\ldots,N\}$, which comes from the partition of unity property of the B-spline basis functions, i.e., $\sum_{j=0}^{N+1} B_{j,3}(z)=1,\forall z\in [a,b]$. Finally, the second and $N+1$-th rows are attributed due to the use of a clamped knot vector. 


The boxed sub matrix $\mathbf B_{\mathrm{rect}}$ consists of rows $2$ to $N+1$ of $\mathbf B$, corresponds to the interpolation conditions at the distinct and ordered evaluation points $\{\alpha_{i}\}_{i=0}^{N-1}$, i.e., $s(\alpha_{i})=f(u(\alpha_{i}))$, and it has full row rank $N$ \cite[Chapter XIII]{z9}. Hence, the interior rows of $\mathbf B$ are linearly independent. The first row of $\mathbf B$ arises from the natural spline boundary condition $s''(\alpha_0)=0$ and has the form $(a_1,-b_1,d_1,0,\ldots,0)$, where $a_1,b_1,d_1>0$ and $a_1-b_1+d_1=0$. In contrast, each interior row is a linear combination of several B-spline basis functions evaluated at $\{\alpha_{i}\}_{i=0}^{N-1}$ and therefore contains only non-negative entries, owing to the non-negativity property of B-spline basis functions i.e., $B_{j,3}(z)\geq 0,\forall z\in[a,b]$. Since the first row contains a strictly negative entry in a column where all interior rows are non-negative, it cannot be expressed as a linear combination of the interior rows and thus linearly independent of the rows of $\mathbf B_{\mathrm{rect}}$. A similar argument applies to the last row of $\mathbf B$, which enforces the boundary condition $s''(\alpha_{N-1})=0$ and has the form $(0,\ldots,0,a_{N+2},-b_{N+2},d_{N+2})$ with a strictly negative entry. Further, the two boundary rows are linearly independent of each other as their nonzero entries are present in the disjoint columns i.e., at near opposite ends of the matrix. Therefore, all $N+2$ rows of $\mathbf B$ are linearly independent, and therefore $\operatorname{rank}(\mathbf B)=N+2$. Since  $\mathbf{B}$ is square matrix of dimension $(N+2)\times(N+2)$, it follows that $\mathbf B$ is invertible. This completes the proof.
\end{proof}


\subsection{Proof of Proposition \ref{lem:lebesgue}}
\label{proof:2}

\begin{proof}
We consider the  BSCC setting, when $S=0$. Let $g(z)$ to refer $f(u(z))$, which denote the composite function. Since, when $S=0$, the master receives function evaluations from all $N$ workers. In particular, worker $i$, for $i=0,1,\ldots,N-1$, evaluates $g(z)$ at the point $\alpha_i$ and returns the value $g(\alpha_i)=f(u(\alpha_i))$. The master reconstructs an approximation of $g(z)$ using a cubic B-spline interpolant as

\begin{equation}
\label{eq:spline_approx}
Tg(z)= \sum_{j=0}^{N+1} c_j B_{j,3}(z),
\end{equation}

\noindent where $T$ is the underlying B-spline interpolation operator that approximates $g(z)$ to $Tg(z)\in\mathcal{S}_{3,\mathbf{t}}$, $\{B_{j,3}(z)\}_{j=0}^{N+1}$ are cubic B-spline basis functions associated with the clamped knot vector $\mathbf{t}$, and
\[\mathbf{c} =\begin{bmatrix}C_{0,3} & C_{1,3} & \cdots & C_{N+1,3}\end{bmatrix}^T\in \mathbb{R}^{(N+2)\times 1}\] \noindent is the vector of spline coefficients. The coefficients $\mathbf{c}$ are determined by the interpolation conditions

\[Tg(\alpha_i)=g(\alpha_i), \quad i=0,1,\ldots,N-1,\]

together with the natural cubic spline boundary conditions
\[Tg''(\alpha_0)=0, \qquad Tg''(\alpha_{N-1})=0.\]

These conditions yield the square linear system
\[\mathbf{B}\mathbf{c} = \bar{\mathbf{f}},\]

\noindent where $\mathbf{B}\in\mathbb{R}^{(N+2)\times(N+2)}$ is the cubic B-spline matrix defined in \eqref{eq:square_bspline_matrix}, and
\[\bar{\mathbf{f}} =
\begin{bmatrix}0 & g(\alpha_0) & g(\alpha_1) & \cdots & g(\alpha_{N-1}) & 0\end{bmatrix}^T\in \mathbb{R}^{(N+2)\times 1}.
\]

Since $\mathbf{B}$ is invertible, hence

\[\mathbf{c} = \mathbf{B}^{-1}\bar{\mathbf{f}}.\] 

Substituting the above expression for $\mathbf{c}$ in \eqref{eq:spline_approx}, we obtain

\begin{equation}
\label{eq:Tg(z) whole exp}
\begin{aligned}
Tg(z)&= \sum_{j=0}^{N+1} c_j B_{j,3}(z) \\
&= \sum_{j=0}^{N+1}
\left(\sum_{k=0}^{N+1} (\mathbf{B}^{-1})_{jk}\,\bar{f}_k\right)B_{j,3}(z),
\end{aligned}
\end{equation}
where $\bar{f}_k$ denoted the $k$-th entry of $\bar{\mathbf{f}}$. Interchanging the order of summation in \eqref{eq:Tg(z) whole exp}, we obtain
\begin{equation}
\label{eq:phi_k_define}
Tg(z)= \sum_{k=0}^{N+1}\left(\sum_{j=0}^{N+1} (\mathbf{B}^{-1})_{jk} B_{j,3}(z)\right)\bar{f}_k.
\end{equation}
We denote the inner sum of \eqref{eq:phi_k_define} by
\begin{equation}
\label{eq:phi_k basis}
 \Phi_k(z):= \sum_{j=0}^{N+1} (\mathbf{B}^{-1})_{jk} B_{j,3}(z),\qquad k=0,\ldots,N+1,
 \end{equation}
 where the functions $\{\Phi_k(z)\}_{k=0}^{N+1}$ constitute the spline basis functions. 
Therefore we express \eqref{eq:phi_k_define}, in terms of $\{\Phi_k(z)\}_{k=0}^{N+1}$ as
 
\begin{equation}
\label{eq:phi_k}
Tg(z) = \sum_{k=0}^{N+1} \Phi_k(z)\,\bar{f}_k.
\end{equation}
Since $\bar{f}_0=\bar{f}_{N+1}=0$ due to the natural boundary conditions, and $\bar{f}_{k}=g(\alpha_{k-1})$ , for $k=1,2,\ldots,N$, \eqref{eq:phi_k} reduces to

\[Tg(z)=\sum_{k=1}^{N} \Phi_k(z)\, g(\alpha_{k-1}).\]

\noindent We are now interested to bound $|Tg(z)|$, which implies
\begin{equation}
\label{eq:||Tg(z)||}
|Tg(z)|\leq\sum_{k=1}^{N} |\Phi_k(z)|\, |g(\alpha_{k-1})|,
\end{equation} 

where $|g(\alpha_{k-1})| \leq \,\|g\|_\infty \qquad \forall k.$

Therefore, we use \eqref{eq:phi_k basis}, and bound $\sum_{k=1}^N |\Phi_k(z)|$
\begin{equation*}
\sum_{k=1}^{N} |\Phi_k(z)|=\sum_{k=1}^{N}\left|\sum_{j=0}^{N+1} (\mathbf{B}^{-1})_{jk} B_{j,3}(z)\right|.
\end{equation*}

Further, we apply the triangular inequality in the above expression, and obtain
\begin{equation*}
\sum_{k=1}^{N} |\Phi_k(z)|\leq
\sum_{k=1}^{N}\sum_{j=0}^{N+1}\left|(\mathbf{B}^{-1})_{jk} B_{j,3}(z)\right|.
\end{equation*}

Now, we use the B-spline basis functions non-negativity property defined in Definition \ref{def:cox de boor} i.e., $B_{j,3}(z)\geq 0$ for all $j\in \{0,1,\ldots,N+1\}$ and $z\in [a,b]$, we obtain
\begin{equation*}\sum_{k=1}^{N} |\Phi_k(z)|
\leq\sum_{k=1}^{N}\sum_{j=0}^{N+1}|(\mathbf{B}^{-1})_{jk}|\, B_{j,3}(z),
\end{equation*}
which implies
\begin{equation*}
\sum_{k=1}^{N} |\Phi_k(z)|
\leq \sum_{j=0}^{N+1}B_{j,3}(z)\sum_{k=1}^{N}|(\mathbf{B}^{-1})_{jk}|,
\end{equation*} and can be upper bounded as
\begin{equation}
\label{eq:final PHI_K}
\leq \bigg(\max_{0\leq j\le N+1}\sum_{k=1}^{N}\big|(\mathbf{B}^{-1})_{jk}\big|\bigg)\sum_{j=0}^{N+1}B_{j,3}(z).
\end{equation}
Now, we apply the partition of unity property of B-spline basis functions defined in Definition  \ref{def:spline_space_bspline} i.e., $\sum_{j=0}^{N+1} B_{j,3}(z)=1$, in \eqref{eq:final PHI_K}, and obtain
\begin{equation}
\label{eq:B_inf}
\sum_{k=1}^{N} |\Phi_k(z)|\leq \max_{0\leq j\leq N+1}\sum_{k=1}^{N}|(\mathbf{B}^{-1})_{jk}|=\|[\mathbf{B}^{-1}]_{\mathrm{sub}}\|_\infty.
\end{equation}
 where $[\mathbf{B}^{-1}]_{\mathrm{sub}}\in\mathbb{R}^{(N+2)\times N}$ denote the sub matrix of $\mathbf{B}^{-1}$ obtained by removing its first and last columns. Finally, we use the bound obtained for $\sum_{k=1}^{N} |\Phi_k(z)|$ in \eqref{eq:B_inf}, and substitute it in \eqref{eq:||Tg(z)||}. Consequently, we get

\begin{equation*}
|Tg(z)|\leq\sum_{k=1}^{N} |\Phi_k(z)||g(\alpha_{k-1})|
\leq\|[\mathbf{B}^{-1}]_{\mathrm{sub}}\|_\infty \|g\|_\infty.
\end{equation*}

\noindent Lastly, we use the definition of operator norm given in Definition \ref{def:operator_norm}, and obtain the bound on operator norm  as \begin{equation*}
\|T\|\leq\|[\mathbf{B}^{-1}]_{\mathrm{sub}}\|_\infty.
\end{equation*}
This completes the proof.
\end{proof}

\subsection{Proof of Theorem \ref{thm:lebegue_bound}}
\label{proof:3}
\begin{proof}
Let
$\mathbf{B}_{\mathrm{rect}} \in \mathbb{R}^{N \times (N+2)}$ be the rectangular matrix defined in \eqref{eq:B_rect} with full row-rank, i.e., $\operatorname{rank}(\mathbf{B}_{\mathrm{rect}})=N.$
Also, let $\mathbf{B}\in\mathbb{R}^{(N+2)\times (N+2)}$ be the square matrix defined in \eqref{eq:square_bspline_matrix} obtained from
$\mathbf{B}_{\mathrm{rect}}$ by appending two boundary rows (one at the top and one at the bottom), such that $\mathbf{B}$ is invertible. Now define the subspace
    \[
    \mathbf{V} \triangleq \mathcal{N}(\mathbf{B}_{\mathrm{rect}})
    = \left\{ \mathbf{x} \in \mathbb{R}^{N+2} : \mathbf{B}_{\mathrm{rect}} \mathbf{x} = 0 \right\}.
    \]
Since \(\mathbf{B}_{\mathrm{rect}}\) is full row-rank, $\dim(\mathbf{V}) = 2.$
Let \(\mathbf{W} \subset \mathbb{R}^{N+2}\) be the complementary subspace of \(\mathbf{V}\), i.e.,
\[
\mathbb{R}^{N+2} = \mathbf{V} \oplus \mathbf{W},
\qquad
\mathbf{V} \cap \mathbf{W} = \{0\}.
\]
Thus, every vector \(\mathbf{c} \in \mathbb{R}^{N+2}\) admits a unique decomposition $\mathbf{c} = \mathbf{v} + \mathbf{w},$ such that $\mathbf{v} \in \mathbf{V},\; \mathbf{w} \in \mathbf{W}.$ As a consequence, the transformation $\mathbf{B}_{\mathrm{rect}}$ when restricted of $\mathbf{W}$ is injective. For the matrix, $\mathbf{B}^{-1}$, let $[\mathbf{B}^{-1}]_{sub}$ denote its submatrix by omitting the first and the last column. For $\mathbf{f}\in\mathbb{R}^{N}$, define the zero-padded extension
\[
\tilde{\mathbf{f}}=[0,\mathbf{f},0]^{\top}\in\mathbb{R}^{N+2},
\]
and set
\[
[\mathbf{B}^{-1}]_{\mathrm{sub}}\mathbf{f}
\triangleq
\mathbf{B}^{-1}\tilde{\mathbf{f}},
\]
which can be verified in a straightforward manner. By construction of the boundary rows, the interior operator satisfies the condition that for every non-zero $\mathbf{f} \in \mathbb{R}^{N}$, 
\begin{equation}
\mathbf{B}_{\mathrm{rect}}\mathbf{B}^{-1}\tilde{\mathbf{f}} = \mathbf{B}_{\mathrm{rect}}[\mathbf{B}^{-1}]_{sub}\mathbf{f} = \mathbf{f}.
\label{eq:interior_identity}
\end{equation} 
Using the above details, we infer that for every non-zero vector \( \mathbf{f} \in \mathbb{R}^{N} \), which corresponds to the evaluations at the knots, the vector $[\mathbf{B}^{-1}]_{sub}\mathbf{f} \in \mathbf{W}$, and moreover, $[\mathbf{B}^{-1}]_{sub}$ is an invertible transformation from $\mathbb{R}^{N}$ to $\mathbf{W}$.\footnote{We restrict our domain to $\mathbf{W}$ since the family of spline functions generated by our interpolation operator have their coefficients belonging to $\mathbf{W}$ owing to $[\mathbf{B}^{-1}]_{sub}$.} Furthermore, let $P_{inv}([\mathbf{B}^{-1}]_{sub})$ represent its psuedo-inverse, which is only defined to act on $\mathbf{W}$. By definition, the minimum singular value of $P_{inv}([\mathbf{B}^{-1}]_{sub})$, when restricted to $\mathbf{W}$, is given by 
\[
\sigma_{\min}(P_{inv}([\mathbf{B}^{-1}]_{sub}))
\;\triangleq\;
\inf_{\mathbf{c}\neq 0, \mathbf{c} \in \mathbf{W}}
\frac{\|P_{inv}([\mathbf{B}^{-1}]_{sub})\mathbf{c}\|_{F}}{\|\mathbf{c}\|_{F}},
\]
which can be alternatively written as
\begin{equation}
\label{eq:singular_value_exp}
\sigma_{\min}(P_{inv}([\mathbf{B}^{-1}]_{sub}))
=
\inf_{\mathbf{c}\neq 0, \mathbf{c} \in \mathbf{W}}
\frac{\|\mathbf{B}_{rect}\mathbf{c}\|_{F}}{\|\mathbf{c}\|_{F}},
\end{equation}
where the above equality follows from \eqref{eq:interior_identity}. Here $\|\cdot\|_{F}$ denotes the Frobenius norm operator.

Using the standard inequality connecting the infinite norm and the Frobenious norm of a matrix, together with the fact that the Frobenius norm of the inverse operator is bounded by the reciprocal of the smallest singular value,
we obtain

\begin{eqnarray}
\label{eq:sing_norm_connection}
\|[\mathbf{B}^{-1}]_{sub}\|_{\infty} & \leq &  
\sqrt{N+2}\|[\mathbf{B}^{-1}]_{sub}\|_{F}, \nonumber\\
& \leq & \frac{N +2}{\sigma_{\min}(P_{inv}([\mathbf{B}^{-1}]_{sub}))},
\end{eqnarray}
wherein the first inequality follows from the well known connection between the infinity norm and Frobenius norm of a matrix, whereas the second inequality is because of the relation between the Frobenius norm of a matrix and the minimum singular value of its inverse. The composition of the two bounds gives $\sqrt{N+2}\sqrt{N}$ in the numerator, which is further bounded as $(N+2)$. 

For every \(\mathbf{c} \in \mathbf{W}\), we now proceed to lower bound the squared norm \(\|\mathbf{B}_{rect}\mathbf{c}\|_{F}\)
using properties of spline interpolation.
To this end, let \( \{t_j\} \) denote the knot sequence associated with the spline basis, and for a spline function $s(z)$ generated by the coefficient
vector \(\mathbf{c}\), define $\mathbf{f}_j \;\triangleq\; s(t_j),$ i.e., $\mathbf{f}_j$ denotes the evaluation of the spline function at the knot
locations. Since the action of the matrix \(\mathbf{B}_{\mathrm{rect}}\) corresponds to
sampling the spline (or its derivatives) at the knot values, the vector
\(\mathbf{B}_{\mathrm{rect}}\mathbf{c}\) can be identified with the collection
of evaluations \(\{\mathbf{f}_j\}\).
Consequently, the squared norm satisfies $\|\mathbf{B}_{\mathrm{rect}}\mathbf{c}\|^{2}_{F} = \sum_{j} |\mathbf{f}_j|^{2}$.

We now use classical spline interpolation results to lower bound the sum of
squared evaluations \(\sum_j |\mathbf{f}_j|^2\) in terms of the mesh parameters. First, the energy of the spline function over its entire support can be lower
bounded by a term, which is a product of constant that depends only on the spline order, the minimum spacing $h_{\min}$ of
the knot sequence, and the squared norm of the coefficient vector used to generate the spline function \cite[Section 2.1]{FEP}. That is, for a spline function \(s(z)\) with coefficient vector
\(\mathbf{c}\), there exists a constant \(C_3>0\), which depends only on the spline
order, and the minimum knot spacing parameter $h_{\min}$, such that
\begin{equation}
\label{eq:spline_energy_lower_bound}
\int_{\mathrm{supp}(s(z))} |s(z)|^2 \, dz \;\ge\; C_3 h_{\min} \, \|\mathbf{c}\|^2_{F}.
\end{equation}

Subsequently, the energy of the spline function over its entire support can be upper bounded in terms of the
discrete evaluations at the knots. In particular, there exists a constant \(C_2>0\), depending
only on the spline order, and the maximum knot spacing $h_{max}$, such that
\begin{equation}
\label{eq:spline_energy_upper_bound}
\int_{\mathrm{supp}(s(z))} |s(z)|^2 \, dz \;\le\; C_2 h_{max} \sum_j |\mathbf{f}_j|^2.
\end{equation}
The above bound can be derived by first partitioning the support set of spline function as a union of consecutive intervals of knots. Then the energy under one interval can be upper bounded by using the discrete evaluation of the function in that interval and the width of the interval. Subsequently, adding the bounds on the energies in each interval, and then  bounding the interval size by the maximum of the intervals, we get the above bound. 
Combining these two results, we obtain a lower bound
\[
\sum_j |\mathbf{f}_j|^2 \;\geq\; \frac{C_{3}h_{min}}{C_{2}h_{max}} \, \|\mathbf{c}\|^2.
\]
Finally, by normalizing the coefficients so that \(\|\mathbf{c}\|^2 = 1\), the
lower bound depends solely on constants determined by the spline order and the knot spacing:
\[
\sum_j |\mathbf{f}_j|^2 \;\ge\; \frac{C_3}{C_2} \frac{h_{min}}{h_{max}}.
\]

Finally, using the above bound in \eqref{eq:singular_value_exp}, we get
\begin{equation*}
\sigma_{\min}(P_{inv}([\mathbf{B}^{-1}]_{sub}))
= \inf_{\mathbf{c}\neq 0, \mathbf{c} \in \mathbf{W}}
\frac{\|\mathbf{B}_{rect}\mathbf{c}\|_{F}}{\|\mathbf{c}\|_{F}} \;\ge\; \frac{C_3}{C_2} \frac{h_{min}}{h_{max}},
\end{equation*}
\noindent and then using \eqref{eq:singular_value_exp} in \eqref{eq:sing_norm_connection}, we get 
\[\|[\mathbf{B}^{-1}]_{sub}\|_{\infty} \;\leq\; (N + 2) \frac{C_2}{C_3} \frac{h_{max}}{h_{min}}.\\ \leq\; {C_1} (N + 2)\frac{h_{max}}{h_{min}},\] where $C_{1}=\frac{C_{2}}{C_{3}}>0$, and is a constant depends on only on the degree of the spline. This completes the proof. 
\end{proof}

\subsection{Proof of Theorem \ref{th:first prelim}}
\label{proof:4}
We use the results from \cite[Chapter XIII, page 181]{z9} to present an upper bound on the approximation error of BSCC using the upper bound on operator norm $\|[\mathbf{B}^{-1}]_{sub}\|_{\infty}$. Hence, this proof is an immediate results of Proposition \ref{lem:lebesgue}, and Theorem~\ref{thm:lebegue_bound}.
\end{appendix}
\end{document}